\documentclass[11pt]{article}
\usepackage[cp1251]{inputenc}
\usepackage[T2A]{fontenc}
\usepackage[english]{babel}
\usepackage{ alphalph,etoolbox }
\usepackage{amsthm, amsmath, amssymb, amsbsy, amscd, amsfonts, latexsym, euscript,epsf, bm}

\usepackage{bbold, bm}
\usepackage{epic,eepic}
\usepackage{texdraw}
\usepackage[dvips]{color}
\usepackage{color}
\usepackage{ dsfont }
\usepackage{graphicx}
\usepackage{exscale,relsize}
\usepackage{authblk}
\usepackage{url}
\usepackage{enumerate}
\usepackage{graphicx}

\newtheorem{theorem}{Theorem}

\newtheorem{proposition}{Proposition}
\newtheorem{corollary}{Corollary}
\newtheorem*{criterion*}{Criterion}
\newtheorem*{definition*}{Definition}
\newtheorem{definition}{Definition}

\theoremstyle{remark} 
\newtheorem{remark}{Remark}

\DeclareMathOperator{\rank}{rank}

\title{Integrable discretisations of the noncommutative NLS equation}
\date{}

\author[1]{S. Konstantinou-Rizos\thanks{skonstantin84@gmail.com}}
\author[2]{P. Xenitidis\thanks{xenitip@hope.ac.uk}}
\affil[1]{Centre of Integrable Systems, P.G. Demidov Yaroslavl State University, Yaroslavl, Russia}
\affil[2]{School of Mathematics and the Environment, Liverpool Hope University, L16 9JD Liverpool, UK}

\patchcmd{\subequations}{\alph{equation}}{\alphalph{\value{equation}}}{}{}

\begin{document}
	
	\maketitle
	
	\begin{abstract}
		We show how to derive noncommutative versions of integrable partial difference equations using Darboux transformations. As an illustrative example, we use the nonlinear Schr\"odinger (NLS) system. We derive a noncommutative nonlinear Schr\"odinger equation and we construct its integrable discretisations via the compatibility condition of Darboux transformations around the square. In particular, we construct a noncommutative Adler--Yamilov type  system and a noncommutative discrete Toda equation. For the  noncommutative Adler--Yamilov type system we construct B\"acklund transformations.
	\end{abstract}
	
	\bigskip
	
	\hspace{.2cm} \textbf{PACS numbers:} 02.30.Ik, 02.90.+p, 03.65.Fd.
	
	\hspace{.2cm} \textbf{Mathematics Subject Classification 2020:} 35Q55, 16T25.
	
	\hspace{.2cm} \textbf{Keywords:} Darboux transformations, B\"acklund transformations, quad-graph equations, non-
	
	\hspace{2.4cm} commutative partial difference equations, noncommutative integrable lattice equations, 
	
	\hspace{2.4cm} noncommutative NLS equation, noncommutative Adler--Yamilov system, noncommu-
	
	\hspace{2.4cm} tative discrete Toda equation.
	
	\section{Introduction}
	Integrable (systems of) equations of Mathematical Physics are of great interest and importance due to their numerous applications, from the propagation of water waves and soliton theory to electromagnetic interactions, optics, and quantum field theory. From the point of view of Mathematical Physics, the solutions to integrable systems deserve attention due to their algebro-geometric properties, as well as their relation to several theories, such as the theory of Hamiltonian dynamical systems and the theory of Yang--Baxter and Zamolodchikov's tetrahedron maps. Moreover, noncommutative versions of integrable equations have become of great interest over the past few decades. For extended details, we refer to \cite{Sokolov} and references therein. Famous examples include the Korteweg-de Vries (KdV) equation, the NLS equation, the sine-Gordon equation, the Boussinesq equation, and other well-known equations of Mathematical Physics. 
	
	In this paper, we aim to generalise to the noncommutative case various results and methods which have been developed in the field of integrable systems during the last decade. More precisely, we extend the ideas of the Darboux-Lax scheme presented in \cite{SPS} for the discretisation of nonlinear integrable partial differential equations (PDEs) employing Darboux transformations, as well as the ideas in \cite{FKRX} for constructing solutions to nonlinear partial difference equations by deriving auto--B\"acklund transformations systematically. We show that the schemes employed in \cite{SPS, FKRX} can be generalised for discretising noncommutative integrable PDEs and also for constructing solutions to noncommutive partial difference equations. As an illustrative example, we use the NLS equation.
	
	The NLS equation is one of the most fundamental equations of Mathematical Physics with a plethora of applications in the study of small-amplitude gravity waves on the surface of deep inviscid waters, in the study of propagation of light in nonlinear optical fibers and many other applications. The most common form of the NLS  equation is the system
	\begin{equation}\label{NLS-eq}
		p_t=\frac{1}{2}p_{xx}-4p^2q,\quad  q_t=-\frac{1}{2}q_{xx}+4pq^2,
	\end{equation}
	where $p=p(x,t)$ and $q=q(x,t)$, $x\in\mathbb{R}$, $t>0$, and indices denote partial derivatives. Regarding its integrability, the NLS system is equivalent to the zero curvature relation ${\rm{U}}_t-{\rm{V}}_x+{\rm{U}}{\rm{V}}-{\rm{V}}{\rm{U}}=0$, where
	\begin{equation}\label{Lax-NLS}
		{\rm{U}}=\lambda\begin{pmatrix}
			1 & 0\\
			0 & -1
		\end{pmatrix}
		+
		\begin{pmatrix}
			0 & 2p\\
			2q & 0
		\end{pmatrix}, \quad
		{\rm{V}}=\lambda^2 \begin{pmatrix}
			1 & 0\\
			0 & -1
		\end{pmatrix}
		+\lambda \begin{pmatrix}
			0 & 2p\\
			2q & 0
		\end{pmatrix}
		+
		\begin{pmatrix}
			-2pq & p_x\\
			-q_x & 2pq
		\end{pmatrix}.
	\end{equation}

	The paper is structured as follows. In the next section, we introduce our notation and give all the necessary definitions for the text to be self-contained. Specifically, we define Darboux and B\"acklund transformations for integrable PDEs, as well as integrable partial difference equations. Furthermore, we explain wha we mean by integrable discretisation in this context. In section \ref{sec-NLS}, we derive a noncommutative NLS system, its Darboux and B\"acklund transformations, and discuss their reductions using first integrals. In section \ref{Int-discr}, we construct a noncommutative integrable discretisation of the NLS system, which can be reduced to a noncommutative Adler--Yamilov type of system using first integrals. Morevover, we construct a noncommutative fully discrete Toda type equation. Section \ref{Darboux-Baecklund} deals with the derivation of Darboux and B\"acklund transformations for the noncommutative discrete integrable Adler--Yamilov type of system. Finally, in section \ref{conclusions}, we close with some concluding remarks and ideas for potential extension of our results and connections with other recent works.

	\section{Preliminaries}\label{preliminaries}
	
	\subsection{Notation}
	In what follows, we deal with differential and difference equations and $x,t \in {\mathbb{R}}$ will denote the continuous independent variables and $n,m\in\mathbb{Z}$ the discrete ones. The dependence of functions on $x$ and $t$ will be omitted, and their dependence on $n$ and $m$ will be denoted with indices, i.e., $u_{ij} = u(n+i,m+j)$. We will use operators $\partial_x$, $\partial_t$ for differentiation with respect to $x$ and $t$, respectively, and denote with $\mathcal{S}$ and $\mathcal{T}$ the two shift operators defined as $\mathcal{S}^i\mathcal{T}^ju(n,m)=u_{ij}$. 
	
	Commutative variables will be denoted with lower case italic letters, e.g., $p$ and $q$. Noncommutative variables will be denoted with bold italic letters; that is, $\bm{p}$ and $\bm{q}$ with $\bm{p}\bm{q}\neq \bm{q}\bm{p}$ in general. Matrices will be denoted with Roman uppercase letters. The dependence of functions on fields and parameters will be stated explicitly when it is necessary. For example, ${\rm{L}}(f_{00},f_{10};a,\lambda)$ denotes a matrix with elements depending on $f_{00}$, $f_{10}$, $a$ and $\lambda$. A semicolon in the arguments of a function will be used to separate variables, e.g., fields $f_{00}$ and $f_{10}$, from the parameters, such as $a$ and $\lambda$. In particular, in what follows $\lambda$ will be used only to denote the spectral parameter in Lax pairs and will be omitted from the arguments of Lax and Darboux matrices.
	
	Finally, by $\mathfrak{R}$ we denote a noncommutative division ring, namely an associative algebra with unity $1$ where all the non-zero elements $\bm{x}$ have a multiplicative inverse $\bm{x}^{-1}$, i.e., $\bm{x}\bm{x}^{-1}=\bm{x}^{-1}\bm{x}=1$. The centre of $\mathfrak{R}$ is denoted by $C(\mathfrak{R}):=\left\{a\in\mathfrak{R}: a\bm{x}=\bm{x}a, \forall \bm{x}\in\mathfrak{R} \right\}$. 
	
	\subsection{Darboux and B\"acklund transformations}
	
	In this section we give the necessary definitions for Lax pairs, Darboux and B\"acklund transformations for continuous and discrete integrable systems. We start with the corresponding definitions for continous systems before moving on to the discrete case.
	
	Let ${\rm{F}}({\bm{u}},{\bm{u}}_t,{\bm{u}}_x,{\bm{u}}_{xx},\ldots) = 0$ be a system of differential equations for ${\bm{u}}=\left({\bm{u}}^{(1)}(x,t),\ldots,{\bm{u}}^{(k)}(x,t)\right)$. The linear system
	\begin{equation} \label{eq:gen-Lax}
		\partial_x \Psi = {\rm{U}}([{\bm{u}}]) \Psi, ~~~  \partial_t \Psi = {\rm{V}}([{\bm{u}}]) \Psi,
	\end{equation}
	where the $N \times N$ matrices ${\rm U}$ and $\rm{V}$ depend on $\bm{u}$ and their derivatives, as well as the spectral parameter $\lambda$, is a Lax pair for the system ${\rm{F}}=0$ if it is consistent provided that $\bm{u}$ is a solution of the system ${\rm{F}}=0$, and vice versa, i.e.,
	\begin{equation} \label{eq:cont-Lax-comb-cond}
		\partial_t{\rm{U}}- \partial_x{\rm{V}}+{\rm{U}}{\rm{V}}-{\rm{V}}{\rm{U}}=0 ~~~ \Longleftrightarrow ~~~  {\rm{F}}({\bm{u}},{\bm{u}}_t,{\bm{u}}_x,{\bm{u}}_{xx},\ldots) = 0.
	\end{equation}
	
	A Darboux transformation for the Lax pair \eqref{eq:gen-Lax} is a transformation that leaves it covariant. More precisely,
	\begin{definition}\label{def-Darboux}
		Let $\bm{u}$ be a solution of system ${\rm{F}}({\bm{u}},{\bm{u}}_t,{\bm{u}}_x,{\bm{u}}_{xx},\ldots) = 0$ and $\Psi$ be the corresponding solution of the Lax pair \eqref{eq:gen-Lax}. If there exists an invertible matrix ${\rm{M}}$, such that
		\begin{subequations}\label{Def-Darboux}
			\begin{equation} \label{Def-Darboux-M}
				\tilde{\Psi} = {\rm{M}}([{\bm{u}}],[\tilde{\bm{u}}];\epsilon) \Psi,
			\end{equation}
			with
			\begin{equation} \label{Def-Darboux-BT}
				\partial_x {\rm{M}} + {\rm{M}} {\rm{U}}([{\bm{u}}]) = {\rm{U}}([\tilde{\bm{u}}]){\rm{M}}, ~~~ \partial_t {\rm{M}} + {\rm{M}} {\rm{V}}([{\bm{u}}]) = {\rm{V}}([\tilde{\bm{u}}]){\rm{M}},
			\end{equation}
		\end{subequations}
		where $\tilde{\bm{u}}$ is a solution of system ${\rm{F}}(\tilde{\bm{u}},\tilde{\bm{u}}_t,\tilde{\bm{u}}_x,\tilde{\bm{u}}_{xx},\ldots) = 0$ and $\tilde{\Psi}$ is the corresponding solution of Lax pair \eqref{eq:gen-Lax} with $\bm{u}$ replaced by $\tilde{\bm{u}}$, then we call \eqref{Def-Darboux} Darboux transformation for Lax pair \eqref{eq:gen-Lax} and matrix ${\rm{M}}([{\bm{u}}],[\tilde{\bm{u}}];\epsilon)$ Darboux matrix. 
	\end{definition}
	
	System \eqref{Def-Darboux-BT} provides us with another integrability aspect of system ${\rm{F}}=0$, namely an auto-B\"acklund transformation, which can be defined as follows.
	
	\begin{definition}
		Let ${\bm{u}}=\left({\bm{u}}^{(1)}(x,t),\ldots,{\bm{u}}^{(k)}(x,t)\right)$ be a solution to the system of PDEs
		\begin{equation} \label{eq:BT-F-eq}
			{\rm{F}}({\bm{u}},{\bm{u}}_t,{\bm{u}}_x,{\bm{u}}_{xx},\ldots) = 0.
		\end{equation}
		Let also ${\bm{u}}$ and ${\bm{v}}=\left({\bm{v}}^{(1)}(x,t),\ldots,{\bm{v}}^{(k)}(x,t)\right)$ satisfy the following system of PDEs
		\begin{equation}\label{BT-c}
			\mathcal{B}_i({\bm{u}},{\bm{u}}_{x},{\bm{u}_{t}},\ldots, {\bm{v}},{\bm{v}}_x,{\bm{v}}_t,\ldots ;\epsilon)=0,\quad i=1,2.
		\end{equation}
		If system $\mathcal{B}_i=0$ can be integrated for ${\bm{u}}$, and the resulting ${\bm{v}}={\bm{v}}(x,t)$ is a solution to 
		\begin{equation} \label{eq:BT-G-eq}
			G({\bm{v}},{\bm{v}}_t,{\bm{v}}_x,{\bm{v}}_{xx},\ldots) =0,
		\end{equation}
		and vice versa, then system \eqref{BT-c} is called a B\"acklund transformation for equations \eqref{eq:BT-F-eq} and \eqref{eq:BT-G-eq}. If ${\rm{F}}\equiv {\rm{G}}$, then \eqref{BT-c} is called an auto-B\"acklund transformation for equation \eqref{eq:BT-F-eq}. Parameter $\epsilon$ in the B\"acklund transformation \eqref{BT-c} is also referred to as the B\"acklund parameter.
	\end{definition}

	Now, we consider integrable systems of quad equations, namely systems of the form
	$${\rm{Q}}({\bm{u}}_{00},{\bm{u}}_{10},{\bm{u}}_{01},{\bm{u}}_{11})=0, ~~~{\mbox{with }}~ {\bm{u}}_{i,j}= \left({\bm{u}}^{(1)}_{i,j},\ldots,{\bm{u}}^{(k)}_{i,j}\right).$$
	The linear system
	\begin{equation}\label{eq:def-dis-Lax-pair}
		{\cal{S}}(\Psi) = {\rm{M}}({\bm{u}}_{00},{\bm{u}}_{10})\Psi, ~~~ {\cal{T}}(\Psi) = {\rm{L}}({\bm{u}}_{00},{\bm{u}}_{01})\Psi,
	\end{equation}
	is a Lax pair for system ${\rm{Q}}=0$ if it is consistent provided that $\bm{u}$ is a solution of the system ${\rm{Q}}= 0$, and vice versa, i.e.,
	\begin{equation}\label{Lax-rep}
		{\rm{M}}(\bm{u}_{01},\bm{u}_{11}){\rm{L}}(\bm{u}_{00},\bm{u}_{01})={\rm{L}}(\bm{u}_{10},\bm{u}_{11}){\rm{M}}(\bm{u}_{00},\bm{u}_{10}) \quad \Longleftrightarrow \quad  {\rm{Q}}(\bm{u}_{00},\bm{u}_{10},\bm{u}_{01},\bm{u}_{11})=0.
	\end{equation}
	The existence of a Lax pair is used as a definition for the integrability of partial difference equations (P$\Delta$Es), see \cite{Hiet-Frank-Joshi} and the references therein. 
	
	We can define Darboux and B\"acklund transformations for integrable quad-graph systems in a way similar to the continuous case.
	
	\begin{definition}\label{Darboux-transform}
		Let $\bm{u}$ be a solution of system ${\rm{Q}}({\bm{u}}_{00},{\bm{u}}_{10},{\bm{u}}_{01},{\bm{u}}_{11}) = 0$ and $\Psi$ be the corresponding solution of the Lax pair \eqref{eq:def-dis-Lax-pair}. If there exists an invertible matrix ${\rm{B}}$, such that
		\begin{subequations}\label{Def-disc-Darboux}
			\begin{equation} \label{Def-disc-Darboux-B}
				\tilde{\Psi} = {\rm{B}}([{\bm{u}}],[\tilde{\bm{u}}];\epsilon) \Psi,
			\end{equation}
			with
			\begin{equation} \label{Def-disc-Darboux-BT}
				{\cal{S}}({\rm{B}}){\rm{M}}(\bm{u}_{00},\bm{u}_{10})={\rm{M}}(\tilde{\bm{u}}_{00},\tilde{\bm{u}}_{10}){\rm{B}}, ~~~  {\cal{T}}({\rm{B}}){\rm{L}}(\bm{u}_{00},\bm{u}_{01})={\rm{L}}(\tilde{\bm{u}}_{00},\tilde{\bm{u}}_{01}){\rm{B}}
			\end{equation}
		\end{subequations}
		where $\tilde{\bm{u}}$ is a solution of system ${\rm{Q}}({\tilde{\bm{u}}}_{00},{\tilde{\bm{u}}}_{10},{\tilde{\bm{u}}}_{01},{\tilde{\bm{u}}}_{11})  = 0$ and $\tilde{\Psi}$ is the corresponding solution of Lax pair \eqref{eq:def-dis-Lax-pair} with $\bm{u}$ replaced by $\tilde{\bm{u}}$, then we call \eqref{Def-disc-Darboux} Darboux transformation for Lax pair \eqref{eq:def-dis-Lax-pair} and matrix ${\rm{B}}([{\bm{u}}],[\tilde{\bm{u}}];\epsilon)$ Darboux matrix.
	\end{definition}
	
	As in the continuous case, system \eqref{Def-disc-Darboux-BT} provides us with an auto-B\"acklund transformation for the discrete system.
	
	\begin{definition}
		Let ${\bm{u}}_{00}=\left({\bm{u}}^{(1)}_{00},\ldots,{\bm{u}}^{(k)}_{00}\right)$ be a solution to the system of difference equations
		\begin{equation} \label{eq:dis-BT-Q}
			{\rm{Q}}({\bm{u}}_{00},{\bm{u}}_{10},{\bm{u}}_{01},{\bm{u}}_{11})=0.
		\end{equation}
		Let also $\bm{u}_{00}$ and $\bm{v}_{00}=\left({\bm{v}}^{(1)}_{00},\ldots,{\bm{v}}^{(k)}_{00}\right)$ satisfy the system of P$\Delta$Es
		\begin{equation}\label{BT}
			\mathcal{B}_i(\bm{u}_{00},\bm{u}_{10},\bm{u}_{01},\bm{u}_{11},\ldots,\bm{v}_{00},\bm{v}_{10},\bm{v}_{01},\bm{v}_{11},\ldots ;\epsilon)=0,\quad i=1,2.
		\end{equation}
		If system $\mathcal{B}_i=0$ can be integrated for $\bm{u}$ and the resulting ${\bm{v}}(n,m)$ is a solution to 
		\begin{equation} \label{eq:dis-BT-P}
			{\rm{P}}(\bm{v}_{00},\bm{v}_{10},\bm{v}_{01})=0,
		\end{equation}
		and vice versa, then system \eqref{BT} is called a B\"acklund transformation for equations \eqref{eq:dis-BT-Q} and \eqref{eq:dis-BT-P}. If $\rm{Q}\equiv\rm{P}$, then \eqref{BT} is called an auto-B\"acklund transformation for equation \eqref{eq:dis-BT-Q}. Parameter $\epsilon$ in \eqref{BT} is also referred to as B\"acklund parameter.
	\end{definition}

	\subsection{Integrable discretisations of PDEs} \label{sec:discretisation}
	
	A connection between integrable PDEs and quad equations is provided via Darboux and auto-B\"acklund transformations and their superposition principle, also known as Bianchi commuting diagram.
	
	Let ${\rm{F}}({\bm{u}},{\bm{u}}_t,{\bm{u}}_x,{\bm{u}}_{xx},\ldots) =0$ be a system of integrable PDEs with Lax pair \eqref{eq:gen-Lax} and two Darboux matrices ${\rm{M}}={\rm{M}}(\bm{u},\bm{u}_{10};\alpha)$ and ${\rm{K}}={\rm{K}}(\bm{u},\bm{u}_{01};\beta)$\footnote{The two Darboux matrices are not necessarily different. Starting with ${\rm{M}}$ one can easily construct the second Darboux matrix ${\rm{K}}$ by replacing $\bm{u}_{10}$ and $\alpha$ with $\bm{u}_{01}$ and $\beta$, respectively, i.e. ${\rm{K}}={\rm{M}}(\bm{u}_{00},\bm{u}_{01};\beta)$.}, where ${\bm{u}}_{10}$ and ${\bm{u}}_{01}$ are the two new solutions of the system and $\alpha$ and $\beta$ are the corresponding B\"acklund parameters.  Hence, we can derive two new solutions from the same ``seed'' solution $\bm{u}$ by employing two Darboux matrices and the corresponding transformations.
	
	Starting with the new solutions ${\bm{u}}_{10}$ and ${\bm{u}}_{01}$ and Darboux matrices ${\rm{K}}_{10}={\rm{K}}(\bm{u}_{10},\bm{u}_{11};\beta)$ and  ${\rm{M}}_{01}={\rm{M}}(\bm{u}_{01},\hat{\bm{u}}_{11};\alpha)$, respectively, we can derive two new solutions $\bm{u}_{11}$ and $\hat{\bm{u}}_{11}$. If we require these two solutions to be equal, $\hat{\bm{u}}_{11} = \bm{u}_{11}$, then the new solution $\bm{u}_{11}$ follows algebraically from the relation ${\rm{K}}_{10} {\rm{M}} = {\rm{M}}_{01} {\rm{K}}$ according to the Bianchi commuting diagram \eqref{bianchi}.
	
	The above procedure is clearly discrete and this allows us to interpret function $\bm{u}$ and auxiliary matrix $\Psi$ of the Lax pair as depending on two discrete variables $n$ and $m$, and rewrite the Darboux transformations as a discrete overdetermined system for $\Psi$ in the following way.
	\begin{equation} \label{eq:DT-Lax}
		\begin{cases}
			&\mathcal{S}({\rm{\Psi}})={\rm{\Psi}_{10}}={{\rm{M}}}(\bm{u},\bm{u}_{10};\alpha){\rm{\Psi}}_{00},\\
			&\mathcal{T}({\rm{\Psi}})={\rm{\Psi}}_{01}={{\rm{K}}}(\bm{u},\bm{u}_{01};\beta){\rm{\Psi}}_{00},
		\end{cases}
	\end{equation}
	where indices now denote shifts in $n$ and $m$. The compatibility condition
	$${{\rm{M}}}(\bm{u}_{01},\bm{u}_{11};\alpha){\rm{K}}(\bm{u},\bm{u}_{01};\beta)={\rm{K}}(\bm{u}_{10},\bm{u}_{11};\beta){\rm{M}}(\bm{u},\bm{u}_{10};\alpha)$$
	of system \eqref{eq:DT-Lax} yields a system of P$\Delta$Es for $\bm{u}$ which is integrable since it possesses system \eqref{eq:DT-Lax} as its Lax pair. We call the resulting system of P$\Delta$Es an \textit{integrable discretisation} of the system of PDEs ${\rm{F}}({\bm{u}},{\bm{u}}_t,{\bm{u}}_x,{\bm{u}}_{xx},\ldots) =0$.

	\begin{figure}[ht]
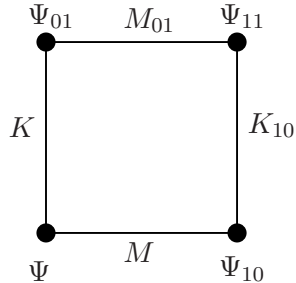

		\centertexdraw{ \setunitscale 0.5
			\linewd 0.02 \arrowheadtype t:F 
			\htext(0 0.5) {\phantom{T}}
			\move (-1 -2) \lvec (1 -2) 
			\move(-1 -2) \lvec (-1 0) \move(1 -2) \lvec (1 0) \move(-1 0) \lvec(1 0)
			\move (1 -2) \fcir f:0.0 r:0.1 \move (-1 -2) \fcir f:0.0 r:0.1
			\move (-1 0) \fcir f:0.0 r:0.1 \move (1 0) \fcir f:0.0 r:0.1  
			\htext (-1.2 -2.5) {$\Psi$} \htext (.8 -2.5) {$\Psi_{10}$} \htext (-.2 -2.3) {$M$}
			\htext (-1.2 .15) {$\Psi_{01}$} \htext (.8 .15) {$\Psi_{11}$} \htext (-.2 .1) {$M_{01}$}
			\htext (-1.4 -1) {$K$} \htext (1.1 -1) {$K_{10}$}}
		\caption{{Bianchi commuting diagram}} \label{bianchi}\index{Bianchi's!commuting diagram.}
	\end{figure}

	\section{A noncommutative NLS equation and its Darboux and B\"acklund transformations}\label{sec-NLS}
	
	In this section, we consider the Lax pair
	\begin{equation}\label{Lax-NLS-NC}
		\begin{array}{l}
			\partial_x \Psi = {\rm{U}}({\bm{p}},{\bm{q}}) \Psi = \left\{\lambda\begin{pmatrix}
				1 & 0\\
				0 & -1
			\end{pmatrix}
			+
			\begin{pmatrix}
				0 & 2\bm{p}\\
				2\bm{q} & 0
			\end{pmatrix} \right\}\Psi, \\
			\\
			\partial_t \Psi =  {\rm{V}}({\bm{p}},{\bm{q}}) \Psi= \left\{ \lambda^2 \begin{pmatrix}
				1 & 0\\
				0 & -1
			\end{pmatrix}
			+\lambda \begin{pmatrix}
				0 & 2\bm{p}\\
				2\bm{q} & 0
			\end{pmatrix}
			+
			\begin{pmatrix}
				-2\bm{p}\bm{q} & \bm{p}_x\\
				-\bm{q}_x & 2\bm{q}\bm{p}
			\end{pmatrix} \right\} \Psi,
		\end{array}
	\end{equation}
	the compatibility condition of which yields the noncommutative system of PDEs
	\begin{equation}\label{NLS-eq-NC}
		\bm{p}_t=\frac{1}{2}\bm{p}_{xx}+4\bm{p}\bm{q}\bm{p}, ~~~ \bm{q}_t=-\frac{1}{2}\bm{q}_{xx}-4\bm{q}\bm{p}\bm{q}.
	\end{equation}
	We refer to this system as the noncommutative NLS system because it becomes system \eqref{NLS-eq} if we replace $\bm{p}$ and $\bm{q}$ with the commutative variables $p$ and $q$.
	
	\begin{remark}
		System \eqref{NLS-eq-NC} was derived in \cite{OS} in the classification of integrable evolution equations on associative algebras, and its first matrix generalisation of the NLS system is due to Manakov \cite{Manakov}.
	\end{remark}
	
	Our aim is to find Darboux matrices and corresponding transformations for the Lax pair \eqref{Lax-NLS-NC}. A Darboux matrix $\rm{M}$ for \eqref{Lax-NLS-NC} can be determined by the relations
	\begin{equation}\label{M-equation}
		\partial_x{\rm M}+{\rm M}{\rm U}-\tilde{\rm{U}} {\rm M}=0,\quad \partial_t{\rm M}+{\rm M}{\rm V}-\tilde{\rm{V}} {\rm M}=0, ~~ {\mbox{with }} \tilde{\rm{U}} = {\rm{U}}(\tilde{\bm{p}},\tilde{\bm{q}}), ~ \tilde{\rm{V}} = {\rm{V}}(\tilde{\bm{p}},\tilde{\bm{q}}),
	\end{equation}
	according to Definition \ref{def-Darboux}. In order to find ${\rm M}$ we need to assume its initial form. Here, we consider only matrices ${\rm M}$ which are linear in the spectral parameter $\lambda$, i.e. ${\rm M}=\lambda {\rm M}^{(1)}+{\rm M}^{(0)}$, and matrix ${\rm M}^{(1)}$ is of rank 1. These considerations lead to the following result.
	
	\begin{theorem}
		Let $\rm{M}=\lambda {\rm M}^{(1)} +{\rm M}^{(0)}$, where $\rank {\rm M}^{(1)}=1$, be a Darboux matrix associated with the Lax pair \eqref{Lax-NLS-NC} of the noncommutative NLS equation \eqref{NLS-eq-NC}. All the Darboux matrices of this form  fall into one of the following two cases.
		\begin{enumerate}
			\item Darboux matrix 
			\begin{subequations} \label{eq:T-10}
				\begin{equation}\label{DT-10}
					{\rm{M}}(\bm{f},\bm{p},\tilde{\bm{q}})=\lambda\begin{pmatrix}
						1 & 0 \\
						0 & 0
					\end{pmatrix}
					+
					\begin{pmatrix}
						\bm{f} & \bm{p} \\
						\tilde{\bm{q}} & 1
					\end{pmatrix},
				\end{equation}
				with $\tilde{\bm{p}}$ and $\tilde{\bm{q}}$ determined by the following system of differential equations.
				\begin{equation}\label{BT-10}
					\bm{f}_x=2\tilde{\bm{p}}\tilde{\bm{q}}-2\bm{p}\bm{q},\quad\bm{p}_x=2\tilde{\bm{p}} -2\bm{f}\bm{p},\quad \tilde{\bm{q}}_x=2\tilde{\bm{q}}\bm{f}-2\bm{q}.
				\end{equation}
			\end{subequations}
			
			\item Darboux matrix
			\begin{subequations} \label{eq:T-01}
				\begin{equation}\label{DT-01}
					{\rm K}(\bm{g},\tilde{\bm{p}},\bm{q})=\lambda\begin{pmatrix}
						0 & 0 \\
						0 & -1
					\end{pmatrix}
					+
					\begin{pmatrix}
						1 & \tilde{\bm{p}} \\
						\bm{q} & \bm{g}
					\end{pmatrix},
				\end{equation}
				with $\tilde{\bm{p}}$ and $\tilde{\bm{q}}$ determined by the following system of differential equations.
				\begin{equation}\label{BT-01}
					\bm{g}_x=2\tilde{\bm{q}}\tilde{\bm{p}}-2\bm{q}\bm{p},\quad \tilde{\bm{p}}_{x}=2\tilde{\bm{p}}\bm{g}-2\bm{p},\quad \bm{q}_x=2\tilde{\bm{q}} -2\bm{g}\bm{q}.
				\end{equation}
			\end{subequations}
			
		\end{enumerate}
		
		There are two degenerated Darboux matrices and corresponding transformations.
		\begin{enumerate}
			\item Darboux matrix
			\begin{subequations} \label{eq:T-10-deg}
				\begin{equation}\label{DT-10-deg}
					{{\rm{M}}}(\bm{f},\bm{p})=\lambda\begin{pmatrix}
						1 & 0 \\
						0 & 0
					\end{pmatrix}
					+
					\begin{pmatrix}
						\bm{f} & \bm{p} \\
						\bm{p}^{-1} & 0
					\end{pmatrix}, 
				\end{equation}
				with $\tilde{\bm{p}}$ and $\tilde{\bm{q}}$ given by
				\begin{equation}\label{BT-10-deg}
					2  \tilde{\bm{p}} = \left( 2 {\bm{p}} {\bm{q}}+ {\bm{f}}_x\right) {\bm{p}},~~~ \tilde{\bm{q}} = {\bm{p}}^{-1}, ~~~  \bm{f}=-\frac{1}{2}\bm{p}_x\bm{p}^{-1}.
				\end{equation}
			\end{subequations}

			\item Darboux matrix
			\begin{subequations} \label{eq:T-01-deg}
				\begin{equation}\label{DT-01-deg}
					{\rm K}(\bm{g},\bm{q})=\lambda\begin{pmatrix}
						0 & 0 \\
						0 & -1
					\end{pmatrix}
					+
					\begin{pmatrix}
						0 &  \bm{q}^{-1} \\
						\bm{q} & \bm{g}
					\end{pmatrix},
				\end{equation}
				with $\tilde{\bm{p}}$ and $\tilde{\bm{q}}$ given by
				\begin{equation}\label{BT-01-deg}
					\tilde{\bm{p}} = {\bm{q}}^{-1}, ~~~ 2 \tilde{\bm{q}}  = \left( 2{\bm{q}}{\bm{p}} + {\bm{g}}_x\right) {\bm{q}},~~~   \bm{g}=-\frac{1}{2}\bm{q}_x\bm{q}^{-1}.
				\end{equation}
			\end{subequations}
		\end{enumerate}
		
	\end{theorem}
	
	\begin{proof}
		Let us write matrices $\rm{U}$, $\tilde{\rm{U}}$, and $\rm{M}$ as
		$$ {\rm{U}} = \lambda {\rm{U}}^{(1)} + {\rm{U}}^{(0)},~~~ \tilde{\rm{U}} = \lambda {\rm{U}}^{(1)} + \tilde{\rm{U}}^{(0)}, ~~\mbox{and}~~ {\rm{M}}=\lambda {\rm{M}}^{(1)} + {\rm{M}}^{(0)} = \lambda\begin{pmatrix}\bm{\alpha} & \bm{\beta}\\ \bm{\gamma} & \bm{\delta} \end{pmatrix}+\begin{pmatrix}
			\bm{f} & \bm{r} \\ \bm{s} & \bm{g} \end{pmatrix},$$
		respectively. Substituting these forms into \eqref{M-equation} and equating the coefficients of the various powers of $\lambda$ to zero, we obtain the following system.
		\begin{subequations} \label{eq:der-DT}
			\begin{align}
				&\lambda^2:\quad  {\rm{M}}^{(1)}\rm{U}^{(1)}-\rm{U}^{(1)}{\rm{M}}^{(1)}=0;\label{l2}\\
				&\lambda^1:\quad {\rm{M}}^{(1)}_x+{\rm{M}}^{(1)}\rm{U}^{(0)}+{\rm{M}}^{(0)}\rm{U}^{(1)}-\tilde{\rm{U}}^{(1)}{\rm{M}}^{(0)}-{\rm{U}^{(0)}}{\rm{M}}^{(1)}=0;\label{l1}\\
				&\lambda^0: \quad {\rm{M}}_x^{(0)}+{\rm{M}}^{(0)}\rm{U}^{(0)}-\tilde{\rm{U}}^{(0)}{\rm{M}}^{(0)}=0.\label{l0}
			\end{align}
		\end{subequations}
		
		Equation \eqref{l2} implies that $\bm{\beta}=\bm{\delta}=0$, and from \eqref{l1} we obtain $\bm{\alpha}_x=\bm{\delta}_x=0$, along with
		\begin{equation} \label{l0b}
			\bm{r}=\bm{\alpha}\bm{p}-\tilde{\bm{p}}\bm{\delta},\quad \bm{s}=\tilde{\bm{q}}\bm{\alpha}-\bm{\delta} \bm{q}.
		\end{equation}
		Since $\rank {{\rm{M}}^{(1)}}=1$, one of $\bm{\alpha}$ or $\bm{\delta}$ must be zero. Consider the case where $\bm{\alpha}\neq 0$ and $\bm{\delta}=0$. Then the relations in \eqref{l0b} reduce to $\bm{r}=\bm{\alpha p}$ and  $\bm{s}=\tilde{\bm{q}}\bm{\alpha}$. Moreover, equation \eqref{l0} implies $\bm{g}_x=0$, i.e. $\bm{g}$ does not depend on $x$. With these choices, the remaining equations in \eqref{eq:der-DT} yield
		\begin{equation} \label{l0c}
			\bm{f}_x=2\tilde{\bm{p}}\tilde{\bm{q}}\bm{\alpha}-2\bm{\alpha} \bm{p}\bm{q},\quad \bm{\alpha} \bm{p}_x=2\tilde{\bm{p}} \bm{g} -2\bm{f}\bm{p},\quad \tilde{\bm{q}}_x \bm{\alpha} =2\tilde{\bm{q}}\bm{f}-2\bm{g q}. 
		\end{equation}
		
		If $\bm{g} \ne 0$, we may, without loss of generality, choose $\bm{\alpha} = \bm{g} = 1$. Indeed, using the scalings 
		$$\bm{f} \rightarrow \bm{f \alpha}, ~~  \bm{p} \rightarrow \bm{\alpha}^{-1} \bm{p g},~~ \bm{q} \rightarrow \bm{g}^{-1} \bm{q \alpha},$$ 
		which leave equation \eqref{NLS-eq-NC} invariant\footnote{System \eqref{NLS-eq-NC} is invariant under the scalings $(\bm{p},\bm{q}) \rightarrow (\bm{a p}, \bm{q} \bm{a}^{-1})$ and $(\bm{p},\bm{q}) \rightarrow (\bm{p b}, \bm{b}^{-1} \bm{q})$, where $\bm{a}$ and $\bm{b}$ are independent of $x$ and $t$.}, we can factor out $\bm{\alpha}$ and $\bm{g}$ in \eqref{l0c}. These considerations yield the Darboux matrix \eqref{DT-10}, and the corresponding transformation \eqref{BT-10}. 
		
		If $\bm{g}=0$, we can set $\bm{\alpha}$ to 1 by applying the scalings 
		$$\bm{f} \rightarrow \bm{f \alpha}, ~~  \bm{p} \rightarrow \bm{\alpha}^{-1} \bm{p}, ~~ \bm{q} \rightarrow \bm{q \alpha}.$$ 
		Then the last two equations in \eqref{l0c} become $\bm{p}_x = -2 \bm{fp}$ and $\tilde{\bm{q}}_{x}=2\tilde{\bm{q}}\bm{f}$. The former equation implies $\bm{f}=-\tfrac{1}{2}\bm{p}_x\bm{p}^{-1}$, in view of which the latter becomes $\tilde{\bm{q}}_{x}\bm{p}+\tilde{\bm{q}}\bm{p}_x=0$, leading to $\tilde{\bm{q}}= \bm{\epsilon}\bm{p}^{-1}$, where $\bm{\epsilon}$ is independent of $x$. Applying the transformation $\tilde{\bm{p}} \rightarrow \tilde{\bm{p}} \bm{\epsilon}^{-1}$, $\tilde{\bm{q}} \rightarrow \bm{\epsilon} \tilde{\bm{q}}$, we can normalise $\bm{\epsilon}$ to $1$, and thus obtain \eqref{eq:T-10-deg}.
		
		In a similar manner, we can treat the case when $\bm{\alpha} = 0$ and $\bm{\delta} \ne 0$. This leads to the two Darboux matrices and their corresponding transformations given in \eqref{eq:T-01} and \eqref{eq:T-01-deg}, respectively.
	\end{proof}
	
	We can simplify equations \eqref{eq:T-10} and \eqref{eq:T-01} further by exploiting the relations \eqref{BT-10} and \eqref{BT-01}, respectively. Specifically, from \eqref{BT-10}, it follows that $\partial_x (\bm{f}-\bm{p}\tilde{\bm{q}}) = 2 [\bm{f},\bm{p}\tilde{\bm{q}}]$, which implies the solution $\bm{f}=\bm{p}\tilde{\bm{q}}+a$, where $a \in C({\mathfrak{R}})$ is independent of $x$. With this choice of $\bm{f}$, the first equation in \eqref{BT-10} becomes an identity, and the remaining two reduce to relations among $\bm{p}$, $\bm{q}$, $\tilde{\bm{p}}$, and $\tilde{\bm{q}}$. 
	
	Similarly, using \eqref{BT-01}, we obtain $\partial_x (\bm{g}-\bm{q}\tilde{\bm{p}}) = 2 [\bm{g},\bm{q}\tilde{\bm{p}}]$, which leads to $\bm{g}=\bm{q}\tilde{\bm{p}}+b$, where $b \in C({\mathfrak{R}})$ is also independent of $x$. With this substitution, the first equation in \eqref{BT-01} becomes an identity, and the remaining two yield relations involving only the fields $\bm{p}$, $\bm{q}$, $\tilde{\bm{p}}$, and $\tilde{\bm{q}}$.
	
	This analysis can be summarised as follows.
	
	\begin{corollary}
		The Darboux transformations \eqref{eq:T-10} and \eqref{eq:T-01} take the following forms.
		\begin{enumerate}
			
			\item Darboux matrix 
			\begin{subequations} \label{eq:DT2-10-all}
				\begin{equation}\label{DT2-10}
					{\rm M}(\bm{p},\tilde{\bm{q}};a) =\lambda\begin{pmatrix}
						1 & 0 \\
						0 & 0
					\end{pmatrix}
					+
					\begin{pmatrix}
						\bm{p}\tilde{\bm{q}} + a & \bm{p} \\
						\tilde{\bm{q}} & 1
					\end{pmatrix},
				\end{equation}
				and $\tilde{\bm{p}}$ and $\tilde{\bm{q}}$ can be determined by the system of differential equations:
				\begin{equation}\label{BT2-10}
					\bm{p}_x=2 \tilde{\bm{p}} - 2 \left(\bm{p} \tilde{\bm{q}} + a\right)\bm{p} ,\quad 
					\tilde{\bm{q}}_x=2 \tilde{\bm{q}} \left( \bm{p}  \tilde{\bm{q}} +a \right)-2\bm{q}.
				\end{equation}
			\end{subequations}
			
			\item Matrix
			\begin{subequations} \label{eq:DT2-01-all}
				\begin{equation}\label{DT2-01}
					{\rm K}(\tilde{\bm{p}}, \bm{q};b)=\lambda\begin{pmatrix}
						0 & 0 \\
						0 & -1
					\end{pmatrix}
					+
					\begin{pmatrix}
						1 & \tilde{\bm{p}} \\
						\bm{q} & \bm{q} \tilde{\bm{p}} + b
					\end{pmatrix},
				\end{equation}
				and $\tilde{\bm{p}}$ and $\tilde{\bm{q}}$ can be determined by the system of differential equations:
				\begin{equation}\label{BT2-01}
					\tilde{\bm{p}}_x= 2\tilde{\bm{p}} (\bm{q}\tilde{\bm{p}}+b)-2\bm{p},\quad \bm{q}_x=2\tilde{\bm{q}} - 2(\bm{q}\tilde{\bm{p}}+b)\bm{q}. 
				\end{equation}
			\end{subequations}
			
		\end{enumerate}
	\end{corollary}
	
	\begin{corollary}
		Relations \eqref{BT2-10} and \eqref{BT2-01} constitute auto-B\"acklund transformations for the noncommutative NLS system \eqref{NLS-eq-NC}.
	\end{corollary}

	\section{Integrable discretisations of the noncommutative NLS equation}\label{Int-discr}
	
	In this section, we employ the Darboux transformations derived in the previous section to construct noncommutative discrete integrable systems, following the approach outlined in Section \ref{sec:discretisation}. Our starting point is the Darboux matrix \eqref{DT-10}, whose consistency around a square, as depicted in the Bianchi diagram in Figure \ref{bianchi}, leads to various discrete integrable systems. By employing first integrals of these systems, we reduce the number of fields, thereby deriving a noncommutative Adler--Yamilov system. We also consider the degenerate Darboux matrix \eqref{DT-10-deg} and the corresponding discrete system, which yields a noncommutative discrete Toda equation.

	\subsection{Integrable discretisation of the noncommutative NLS equation}
	
	We start with the discrete Lax pair
	$${\cal{S}}(\Psi) = {\rm{M}}(\bm{f},\bm{p},\bm{q}_{10}) \Psi, ~~~ {\cal{T}}(\Psi) = {\rm{M}}(\bm{g},\bm{p},\bm{q}_{01}) \Psi,$$
	where matrix $\rm{M}$ is given in \eqref{DT-10}, and consider its compatibility condition
	\begin{equation}\label{NLS-dis-Lax}
		{\rm{M}}(\bm{f}_{01},\bm{p}_{01},\bm{q}_{11}){\rm{M}}(\bm{g},\bm{p},\bm{q}_{01})={\rm{M}}(\bm{g}_{10},\bm{p}_{10},\bm{q}_{11}){\rm{M}}(\bm{f},\bm{p},\bm{q}_{10}).
	\end{equation}
	Writing out explicitly the compatibility condition and equating the coefficients of the different powers of $\lambda$ to zero, we end up with the following system of equations:
	\begin{subequations}\label{NLS-dis-sys}
		\begin{align}
			& \bm{f}_{01}+\bm{g}=\bm{g}_{10}+\bm{f},\label{NLS-dis-sys-a}\\
			& \bm{f}_{01}\bm{g}+\bm{p}_{01}\bm{q}_{01}=\bm{g}_{10}\bm{f}+\bm{p}_{10}\bm{q}_{10},\label{NLS-dis-sys-b}\\
			& \bm{f}_{01}\bm{p}+\bm{p}_{01}=\bm{g}_{10}\bm{p}+\bm{p}_{10},\label{NLS-dis-sys-c}\\
			& \bm{q}_{11}\bm{g}+\bm{q}_{01}=\bm{q}_{11}\bm{f}+\bm{q}_{10}\label{NLS-dis-sys-d}.
		\end{align}
	\end{subequations}
	We can solve system \eqref{NLS-dis-sys} for $(\bm{f}_{01},\bm{g},\bm{p}_{01},\bm{q}_{01})$. Indeed, from \eqref{NLS-dis-sys-a}, \eqref{NLS-dis-sys-c}, and \eqref{NLS-dis-sys-d}, we obtain 
	$$\bm{f}_{01}=\bm{g}_{10}+\bm{f}-\bm{g}, ~~~ \bm{p}_{01}=\bm{g}_{10}\bm{p}+\bm{p}_{10}-\bm{f}_{01}\bm{p}, ~~~ {\mbox{and}} ~~~ \bm{q}_{01}=\bm{q}_{11}\bm{f}+\bm{q}_{10}-\bm{q}_{11}\bm{g},$$ 
	respectively. Their substitution into \eqref{NLS-dis-sys-b} leads to
	\begin{equation} \label{all-f-g}
		(\bm{g}-\bm{f}) \bm{p} \bm{q}_{11} (\bm{f}-\bm{g}) - (\bm{g}-\bm{f}) (\bm{g}-\bm{p} \bm{q}_{10}) + (\bm{g}_{10}-\bm{p}_{10}\bm{q}_{11}) (\bm{g}-\bm{f})=0 
	\end{equation}
	
	If $\bm{g}=\bm{f}$, then this equation holds identically. Moreover, from \eqref{NLS-dis-sys-a} we obtain $\bm{f}_{01}=\bm{g}_{10}$, and then from \eqref{NLS-dis-sys-c} and \eqref{NLS-dis-sys-d}, respectively, we get $\bm{p}_{01}=\bm{p}_{10}$ and $\bm{q}_{01}=\bm{q}_{10}$, namely the trivial solution.
	
	From now on, we assume that $\bm{g} \ne \bm{f}$ and we consider two inequivalent cases.
	\begin{itemize}
		\item[Case 1.] If $[\bm{g}-\bm{f},\bm{g}-\bm{p} \bm{q}_{10}] =0$, then \eqref{all-f-g} becomes
		$$\left((\bm{f}-\bm{g}) \bm{p} \bm{q}_{11} -  (\bm{g}-\bm{p} \bm{q}_{10}) + (\bm{g}_{10}-\bm{p}_{10}\bm{q}_{11})\right) (\bm{g}-\bm{f})=0.$$
		Assuming that the first factor is zero, we end up with
		\begin{equation}\label{f-g}
			\bm{g} = \left(\bm{f} \bm{p} \bm{q}_{11} + \bm{p} \bm{q}_{10} + \bm{g}_{10}-\bm{p}_{10}\bm{q}_{11}\right)(1+\bm{p} \bm{q}_{11})^{-1}.
		\end{equation}
		
		\item[Case 2.] If $[\bm{g}-\bm{f},\bm{g}_{10}-\bm{p}_{10} \bm{q}_{11}] =0$, then \eqref{all-f-g} becomes
		$$ (\bm{g}-\bm{f}) \left(\bm{p} \bm{q}_{11} (\bm{f}-\bm{g}) - (\bm{g}-\bm{p} \bm{q}_{10}) + (\bm{g}_{10}-\bm{p}_{10}\bm{q}_{11}) \right)=0.$$
		The assumption now that the second factor is identically zero leads to
		\begin{equation}\label{g-f} 
			\bm{g} = (1+ \bm{p} \bm{q}_{11})^{-1} \left(\bm{p} \bm{q}_{11} \bm{f} + \bm{p} \bm{q}_{10} + \bm{g}_{10}-\bm{p}_{10}\bm{q}_{11}\right).
		\end{equation}
		
	\end{itemize}
	
	We can now use equation \eqref{f-g} or \eqref{g-f}, along with \eqref{NLS-dis-sys-a}, \eqref{NLS-dis-sys-c} and \eqref{NLS-dis-sys-d}, to express $\bm{f}_{01}$, $\bm{p}_{01}$ and $\bm{q}_{01}$ in terms of $\bm{f}$, $\bm{g}_{10}$, $\bm{p}$, $\bm{p}_{10}$, $\bm{q}_{10}$, and $\bm{q}_{11}$.
	This analysis can be summarised in the following statement.
	
	\begin{proposition} \label{prop:NLS-2-sys}
		Consider the Lax pair
		$${\cal{S}}(\Psi) = {\rm{M}}(\bm{f},\bm{p},\bm{q}_{10}) \Psi, ~~~ {\cal{T}}(\Psi) = {\rm{M}}(\bm{g},\bm{p},\bm{q}_{01}) \Psi,$$
		where
		$$ {\rm{M}}(\bm{x},\bm{y},\bm{z}) = \lambda \begin{pmatrix} 1 & 0 \\ 0 & 0 \end{pmatrix} + 
		\begin{pmatrix} {\bm{x}} & \bm{y} \\ \bm{z} & 1 \end{pmatrix}.$$
		If $[\bm{g}-\bm{f},\bm{g}-\bm{p} \bm{q}_{10}] =0$, then its compatibility condition leads to the system of P$\Delta$Es
		\begin{subequations}\label{NLS-dis1}
			\begin{align}
				& \bm{f}_{01}=\left( \bm{f} - \bm{p} \bm{q}_{10} + (\bm{g}_{10} \bm{p}+ \bm{p}_{10}) \bm{q}_{11} \right) (1+\bm{p} \bm{q}_{11})^{-1},\label{NLS-dis1-a}\\
				& \bm{g}= \left(\bm{f} \bm{p} \bm{q}_{11} + \bm{p} \bm{q}_{10} + \bm{g}_{10}-\bm{p}_{10}\bm{q}_{11}\right)(1+\bm{p} \bm{q}_{11})^{-1},\label{NLS-dis1-b}\\
				& \bm{p}_{01}=(\bm{p} \bm{q}_{10} \bm{p} + (\bm{g}_{10}-\bm{f}) \bm{p} + \bm{p}_{10}) (1+\bm{q}_{11} \bm{p})^{-1},\label{NLS-dis1-c}\\
				& \bm{q}_{01}= (\bm{q}_{11} \bm{p}_{10} \bm{q}_{11} + \bm{q}_{11} (\bm{f} -\bm{g}_{10}) + \bm{q}_{10} + \bm{q}_{10} \bm{p} \bm{q}_{11}-  \bm{q}_{11} \bm{p} \bm{q}_{10}  ) (1+ \bm{p}\bm{q}_{11})^{-1},\label{NLS-dis1-d}.
			\end{align}
		\end{subequations}
		Assuming that $[\bm{g}-\bm{f},\bm{g}_{10}-\bm{p}_{10} \bm{q}_{11}] =0$, then the compatibility condition yields
		\begin{subequations}\label{NLS-dis}
			\begin{align}
				& \bm{f}_{01}=(1+\bm{p}\bm{q}_{11})^{-1}( \bm{f} - \bm{p} \bm{q}_{10} + \bm{p} \bm{q}_{11} \bm{g}_{10} + \bm{p}_{10} \bm{q}_{11}),\label{NLS-dis-a}\\
				& \bm{g}=(1+\bm{p}\bm{q}_{11})^{-1}(\bm{g}_{10}+\bm{p}\bm{q}_{10}-\bm{p}_{10}\bm{q}_{11}+\bm{p}\bm{q}_{11}\bm{f}),\label{NLS-dis-b}\\
				& \bm{p}_{01}=(1+\bm{p}\bm{q}_{11})^{-1}( \bm{p} \bm{q}_{10} \bm{p} + (\bm{g}_{10} -\bm{f}) \bm{p} + \bm{p}_{10} + \bm{p} \bm{q}_{11} \bm{p}_{10} - \bm{p}_{10} \bm{q}_{11} \bm{p}),\label{NLS-dis-c}\\
				& \bm{q}_{01}=(1+\bm{q}_{11} \bm{p})^{-1}( \bm{q}_{11} \bm{p}_{10}\bm{q}_{11} + \bm{q}_{11} (\bm{f}-\bm{g}_{10}) + \bm{q}_{10}).\label{NLS-dis-d}
			\end{align}
		\end{subequations}
		Both systems \eqref{NLS-dis1} and \eqref{NLS-dis} are integrable and admit the conservation law 
		\begin{equation}\label{conservation_law}
			(\mathcal{T}-1)\bm{f}=(\mathcal{S}-1)\bm{g}.
		\end{equation} 
	\end{proposition}
	
	\begin{remark}
		Systems \eqref{NLS-dis1}, \eqref{NLS-dis} are integrable discretisations of the NLS system \eqref{NLS-eq-NC}. Moreover, they may be considered as noncommutative versions of system (18) in \cite{SPS}, because they reduce to the latter system if all variables commute with each other. 
	\end{remark}

	\begin{remark}
		One can consider the initial value problem on the staircase (see Figure \ref{fig-ivp}). Then, system \eqref{NLS-dis} determines the evolution in the NE direction of the staircase. Moreover, system \eqref{NLS-dis-sys} can also be solved for $(\bm{f}_{01},\bm{g},\bm{p}_{01},\bm{q}_{01})$, and this solution determines the evolution in the SW direction of the staircase.
	\end{remark}
	
	\begin{figure}[ht!]
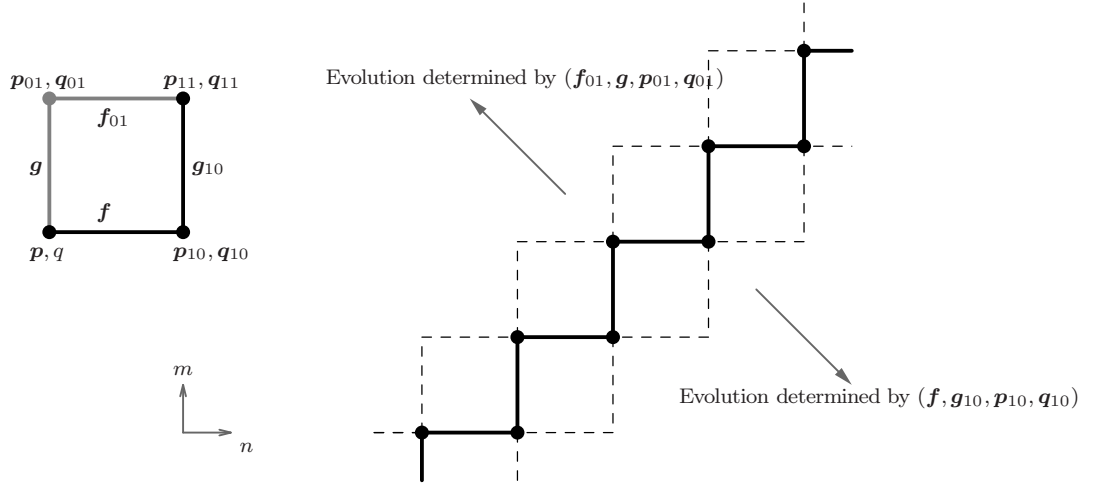

		\centertexdraw{
			\setunitscale 0.5
			\move(-4.5 -2) \linewd 0.02 \setgray 0.4 \arrowheadtype t:V \arrowheadsize l:.12 w:.06 \avec(-4.5 -1.5) 
			\move(-4.5 -2) \arrowheadtype t:V  \avec(-4 -2)
			\arrowheadsize l:.20 w:.10
			\move(-.5 .5) \linewd 0.02 \setgray 0.4 \arrowheadtype t:F \avec(-1.5 1.5) 
			\move(1.5 -.5) \linewd 0.02 \setgray 0.4 \arrowheadtype t:F \avec(2.5 -1.5) 
			\setgray 0.5
			\linewd 0.04 \move (-4.5 1.5)  \lvec (-5.9 1.5) \lvec (-5.9 .1)
			\move (-5.9 1.5) \fcir f:0.5 r:0.075
			\htext (-6.1 .7) {\scriptsize{$\bm{g}$}}
			\htext (-4.4 .7) {\scriptsize{$\bm{g}_{10}$}}
			\setgray 0.0
			\linewd 0.04 \move (-2 -2.5) \lvec (-2 -2) \lvec (-1 -2) \lvec (-1 -1) \lvec (0 -1) \lvec (0 0) \lvec (1 0) \lvec(1 1) \lvec (2 1) \lvec(2 2) \lvec(2.5 2)
			\move (-5.9 .1) \lvec (-4.5 .1) \lvec (-4.5 1.5)  
			\linewd 0.015 \lpatt (.1 .1 ) \move (-2 -2) \lvec (-2 -1) \lvec(-1 -1) \lvec (-1 0) \lvec (0 0) \lvec (0 1) \lvec(1 1) \lvec (1 2) \lvec (2 2) \lvec (2 2.5)
			\move(-2.5 -2) \lvec(-2 -2) \move(-2.5 -2) \lvec(-2 -2)
			\move (-1 -2.5) \lvec (-1 -2) \lvec(0 -2) \lvec(0 -1) \lvec(1 -1) \lvec(1 0) \lvec(2 0) \lvec(2 1) \lvec(2.5 1)
			\move (-2 -2) \fcir f:0.0 r:0.075 \move (-1 -2) \fcir f:0.0 r:0.075
			\move (-1 -1) \fcir f:0.0 r:0.075 \move (0 -1) \fcir f:0.0 r:0.075
			\move (0 0) \fcir f:0.0 r:0.075 \move (1 0) \fcir f:0.0 r:0.075  
			\move (1 1) \fcir f:0.0 r:0.075 \move (2 1) \fcir f:0.0 r:0.075
			\move (2 2) \fcir f:0.0 r:0.075
			\move (-5.9 .1) \fcir f:0.0 r:0.075 \move (-4.5 .1) \fcir f:0.0 r:0.075 \move (-4.5 1.5) \fcir f:0.0 r:0.075
			\htext (-3.9 -2.2) {\scriptsize{$n$}}
			\htext (-4.6 -1.4) {\scriptsize{$m$}}
			\htext (-6.1 -.2) {\scriptsize{$\bm{p},q$}}
			\htext (-4.6 -.2) {\scriptsize{$\bm{p}_{10},\bm{q}_{10}$}}
			\htext (-5.4 0.2) {\scriptsize{$\bm{f}$}}
			\htext (-5.4 1.2) {\scriptsize{$\bm{f}_{01}$}}
			\htext (-6.3 1.6) {\scriptsize{$\bm{p}_{01},\bm{q}_{01}$}}
			\htext (-4.7 1.6) {\scriptsize{$\bm{p}_{11},\bm{q}_{11}$}}
			\htext (-3 1.6) {{\scriptsize{Evolution determined by $(\bm{f}_{01},\bm{g},\bm{p}_{01},\bm{q}_{01})$}}}
			\htext (.7 -1.75) {{\scriptsize{Evolution determined by $(\bm{f},\bm{g}_{10},\bm{p}_{10},\bm{q}_{10})$}}}
		}
		\caption{{Initial value problem and direction of evolution.}} \label{fig-ivp}
	\end{figure}
	
	System \eqref{NLS-dis1} admits a first integral, and a second one exists under certain commutation conditions. Similar considerations apply to system \eqref{NLS-dis}. These results are summarised in the following two statements.
	
	\begin{proposition}\label{1st-ints-1} 
		System \eqref{NLS-dis1} admits the first integral
		\begin{equation}\label{first-ints-1a}
			({\cal{T}}-1)\left(\bm{f}-\bm{p}\bm{q}_{10}\right) = 0.
		\end{equation}
		If additionally $[\bm{f},\bm{p} \bm{q}_{11}] = [\bm{g},\bm{p} \bm{q}_{11}] =0$, then it admits one more first integral, namely
		\begin{equation}\label{first-ints-1b}
			({\cal{S}}-1)\left(\bm{g}-\bm{p}\bm{q}_{01}\right) = 0.
		\end{equation}
	\end{proposition}
	
	\begin{proof}
		We start with equation \eqref{NLS-dis1-a} in the form
		$$\bm{f}_{01}  (1+\bm{p} \bm{q}_{11}) = \bm{f} - \bm{p} \bm{q}_{10} + (\bm{g}_{10} \bm{p}+ \bm{p}_{10}) \bm{q}_{11},$$
		which can be arranged and written as
		$$\bm{f}_{01}- \bm{p}_{01} \bm{q}_{11} = \bm{f} - \bm{p} \bm{q}_{10} + \left(-\bm{p}_{01} -\bm{f}_{01} \bm{p} + \bm{g}_{10} \bm{p}+ \bm{p}_{10}\right) \bm{q}_{11}.$$
		But the last term on the right hand side is identically zero because of equation \eqref{NLS-dis-sys-c}, which proves that $\bm{f}-\bm{p}\bm{q}_{10}$ is a first integral.
		
		Starting with equation \eqref{NLS-dis1-b} and working is the same way, we end up with
		$$\bm{g} -\bm{p} \bm{q}_{01} = \bm{g}_{10} - \bm{p}_{10} \bm{q}_{11} + \bm{p}\bm{q}_{10} -\bm{p}\bm{q}_{01}+ \bm{f} \bm{p} \bm{q}_{11} - \bm{g} \bm{p} \bm{q}_{11}.$$
		The assumption that $[\bm{f},\bm{p} \bm{q}_{11}] = [\bm{g},\bm{p} \bm{q}_{11}] =0$ allows us to write it as
		$$\bm{g} -\bm{p} \bm{q}_{01} = \bm{g}_{10} - \bm{p}_{10} \bm{q}_{11} + \bm{p}\left(\bm{q}_{10} -\bm{q}_{01}+  \bm{q}_{11} \bm{f} -  \bm{q}_{11}\bm{g} \right).$$
		But the last term on the right hand side is identically zero because of equation \eqref{NLS-dis-sys-d}, which proves that $\bm{g}-\bm{p}\bm{q}_{01}$ is another first integral under the assumption that $[\bm{f},\bm{p} \bm{q}_{11}] = [\bm{g},\bm{p} \bm{q}_{11}] =0$.
	\end{proof}
	
	In a similar fashion, we can prove the following.
	
	\begin{proposition}\label{1st-ints-2} 
		System \eqref{NLS-dis} admits the first integral
		\begin{equation}\label{first-ints-2a}
			({\cal{S}}-1)\left(\bm{g}-\bm{p}\bm{q}_{01}\right) = 0.
		\end{equation}
		If additionally $[\bm{f}_{01},\bm{p} \bm{q}_{11}] = [\bm{g}_{10},\bm{p} \bm{q}_{11}] =0$, then it admits one more first integral, namely
		\begin{equation}\label{first-ints-2b}
			({\cal{T}}-1)\left(\bm{f}-\bm{p}\bm{q}_{10}\right) = 0.
		\end{equation}
	\end{proposition}

	\subsection{A noncommutative Adler--Yamilov type of system}
	
	Using the first integrals \eqref{first-ints-1a} and \eqref{first-ints-2a} we can reduce both systems \eqref{NLS-dis1} and \eqref{NLS-dis} to the same integrable system of P$\Delta$Es. The first integrals \eqref{first-ints-1a} and \eqref{first-ints-2a} imply that $\bm{f}-\bm{p}\bm{q}_{10}=a(n)$ and $\bm{g}-\bm{p}\bm{q}_{01}=b(m)$, respectively, with $\alpha(n), \beta(m) \in C({\mathfrak{R}})$. In view of these choices, the commutation relations in Proposition \ref{prop:NLS-2-sys}, i.e., $[\bm{g}-\bm{f},\bm{g}-\bm{p} \bm{q}_{10}] =0$, $[\bm{g}-\bm{f},\bm{g}_{10}-\bm{p}_{10} \bm{q}_{11}] =0$, hold identically. Moreover, equations \eqref{NLS-dis1-c} and \eqref{NLS-dis-c} reduce to
	$$\bm{p}_{01} =\bm{p}_{10}-\left(a(n)-b(m)\right)(1+\bm{p}\bm{q}_{11})^{-1}\bm{p} =\bm{p}_{10}-\left(a(n)- b(m)\right)\bm{p} (1+\bm{q}_{11}\bm{p})^{-1},$$
	and equations \eqref{NLS-dis1-d} and \eqref{NLS-dis-d} become
	$$\bm{q}_{01} =\bm{q}_{10}+\left(a(n)-b(m)\right)\bm{q}_{11} (1+\bm{p}\bm{q}_{11})^{-1} =\bm{q}_{10}+\left(a(n)- b(m)\right)(1+\bm{q}_{11}\bm{p})^{-1}\bm{q}_{11} .$$
	It should be noted that these relations hold in view of the identity $\bm{x} (1+\bm{y} \bm{x})^{-1} = (1+\bm{x} \bm{y})^{-1} \bm{x}$ which holds for any $\bm{x},\bm{y} \in {\mathfrak{R}}$. Finally, our choices for $\bm{f}$ and $\bm{g}$, along with the above equations, turn the remaining equations of systems \eqref{NLS-dis1} and \eqref{NLS-dis} into identities. We can summarise this analysis in the following statement.
	
	\begin{theorem}
		Let $a(n), b(m)\in C(\mathfrak{R})$. The noncommutative discrete system
		\begin{subequations}\label{nc-Adler-Yamilov-r}
			\begin{align}
				\bm{p}_{01}&=\bm{p}_{10}-\left(a(n)- b(m)\right)(1+\bm{p}\bm{q}_{11})^{-1}\bm{p} =\bm{p}_{10}-\left(a(n)- b(m)\right)\bm{p} (1+\bm{q}_{11}\bm{p})^{-1},\label{nc-Adler-Yamilov-r-a}\\
				\bm{q}_{01}&=\bm{q}_{10}+\left(a(n)- b(m)\right)\bm{q}_{11} (1+\bm{p}\bm{q}_{11})^{-1} =\bm{q}_{10}+\left(a(n)- b(m)\right)(1+\bm{q}_{11}\bm{p})^{-1}\bm{q}_{11},\label{nc-Adler-Yamilov-r-b}
			\end{align}
		\end{subequations}
		is integrable with Lax representation
		\begin{subequations} \label{nc-Adler-Yamilov-Lax}
			\begin{equation}
				{\rm M}(\bm{p}_{01},\bm{q}_{11};a(n)){\rm M}(\bm{p},\bm{q}_{01};b(m))={\rm M}(\bm{p}_{10},\bm{q}_{11};b(m)){\rm M}(\bm{p},\bm{q}_{10};a(n)),
			\end{equation}
			where 
			\begin{equation} \label{eq:matM-AY}
				{\rm M}(\bm{x},\bm{y};a) =\lambda\begin{pmatrix} 
					1 & 0 \\
					0 & 0
				\end{pmatrix}
				+
				\begin{pmatrix}
					a+\bm{x}\bm{y} & \bm{x} \\
					\bm{y} & 1
				\end{pmatrix}.
			\end{equation}
		\end{subequations}
	\end{theorem}
	
	\begin{remark}
		System \eqref{nc-Adler-Yamilov-r} constitutes a noncommutative version of the Adler--Yamilov (discrete NLS) system derived in \cite{SPS}.
	\end{remark}
	
	\begin{corollary}
		The noncommutative Adler--Yamilov type system \eqref{nc-Adler-Yamilov-r} constitutes an integrable discretisation of the noncommutative NLS system \eqref{NLS-eq-NC}.
	\end{corollary}

	\subsection{A noncommutative discrete Toda equation}
	
	In this subsection we consider the Lax pair
	$${\cal{S}}(\Psi) = {\rm{K}}(\bm{f},\bm{p}) \Psi, ~~~ {\cal{T}}(\Psi) = {\rm{M}}(\bm{p},\bm{q}_{01};b(m)) \Psi,$$
	where matrix $\rm{K}$ is defined in \eqref{DT-10-deg}, and matrix $\rm{M}$ is given in \eqref{eq:matM-AY}. The compatibility condition of this Lax pair,
	\begin{equation}\label{noncomm-d-Toda}
		{{\rm{K}}}(\bm{f}_{01},\bm{p}_{01}){{\rm{M}}}(\bm{p},\bm{q}_{01};b(m))={{\rm{M}}}(\bm{p}_{10},\bm{q}_{11};b(m)){{\rm{K}}}(\bm{f},\bm{p}),
	\end{equation}
	yields the following system of equations.
	\begin{subequations}\label{Toda-sys}
		\begin{align}
			&\bm{f}_{01}+\bm{p}\bm{q}_{01}=\bm{f}+\bm{p}_{10}\bm{q}_{11},\label{Toda-sys-a}\\
			&\bm{f}_{01}\left[b(m)+\bm{p}\bm{q}_{01}\right]+\bm{p}_{01}\bm{q}_{01}=\left[b(m)+\bm{p}_{10}\bm{q}_{11}\right]\bm{f}+\bm{p}_{10}\bm{p}^{-1},\label{Toda-sys-b}\\
			& \bm{f}_{01}\bm{p}+\bm{p}_{01}=\left[b(m)+\bm{p}_{10}\bm{q}_{11}\right]\bm{p},\label{Toda-sys-c}\\
			& \bm{p}_{01}^{-1}\left[b(m)+\bm{p}\bm{q}_{01}\right]=\bm{q}_{11}\bm{f}+\bm{p}^{-1},\label{Toda-sys-d}\\
			& \bm{p}_{01}^{-1}=\bm{q}_{11}.\label{Toda-sys-e}
		\end{align}
	\end{subequations}
	The last equation readily implies that $\bm{p} = \bm{q}_{10}^{-1}$, in view of which equations \eqref{Toda-sys-c} and \eqref{Toda-sys-d} yield
	\begin{equation}\label{f01}
		\bm{f}_{01}=b(m)+\bm{q}_{20}^{-1}\bm{q}_{11}-\bm{q}_{11}^{-1}\bm{q}_{10}
	\end{equation}
	and
	\begin{equation}\label{f}
		\bm{f}=b(m)+\bm{q}_{10}^{-1}\bm{q}_{01}-\bm{q}_{11}^{-1}\bm{q}_{10},
	\end{equation}
	respectively. Upon the same substitution and relations \eqref{f01} and \eqref{f}, the remaining equations of system \eqref{Toda-sys} become identities. Moreover, the compatibility condition ${\cal{T}}(\bm{f}) = \bm{f}_{01}$ of the above relations leads to the following equation.
	\begin{equation}\label{noncomm-d-Toda}
		\bm{q}_{20}^{-1}\bm{q}_{11}-\bm{q}_{11}^{-1}\bm{q}_{02}+\bm{q}_{12}^{-1}\bm{q}_{11}-\bm{q}_{11}^{-1}\bm{q}_{10}=b(m+1)-b(m).
	\end{equation} 
	If we consider that all $\bm{q}_{i,j}$, $i,j=0,1,2$, commute with each other and set $\bm{q}=\exp{(-w_{-1,-1})}$, then equation \eqref{noncomm-d-Toda} becomes 
	\begin{equation}\label{d-Toda}
		\exp{(w_{1,-1}-w)} -  \exp{(w-w_{-1,1})} + \exp{(w_{01}-w)} - \exp{(w-w_{0,-1})} =  b(m+1) - b(m),
	\end{equation}
	which is the fully discrete Toda equation \cite{Hirota, Suris}. Thus, equation \eqref{noncomm-d-Toda} may be regarded as a noncommutative version of the discrete Toda equation \eqref{d-Toda}.

	\section{B\"acklund transformations for the noncommutative Adler--Yamilov type of system}\label{Darboux-Baecklund}
	
	In this section, we construct Darboux and B\"acklund transformations for the noncommutative Adler--Yamilov system \eqref{nc-Adler-Yamilov-r}. Specifically, we construct a Darboux matrix $\rm{B}$ such that
	\begin{subequations}\label{BM}
		\begin{align}
			& {\cal{S}}\left({\rm{B}}\right){{\rm{M}}}(\bm{p},\bm{q}_{10};a)= {{\rm{M}}}(\tilde{\bm{p}},\tilde{\bm{q}}_{10};a){\rm{B}},\label{BM-a}\\
			& {\cal{T}}\left({\rm{B}}\right){{\rm{M}}}(\bm{p},\bm{q}_{01};b)= {{\rm{M}}}(\tilde{\bm{p}},\tilde{\bm{q}}_{01};b){\rm{B}},\label{BM-b}
		\end{align}
	\end{subequations}
	with matrix $\rm{M}$ given in \eqref{eq:matM-AY}. The simplest Ansatz for the Darboux matrix $\rm{B}$ is to depend linearly on the spectral parameter $\lambda$, i.e., 
	$$\rm{B}=\lambda {\rm B}^{(1)} +{\rm B}^{(0)} =\lambda \begin{pmatrix}
		\bm{\alpha} & \bm{\beta} \\
		\bm{\gamma} & \bm{\delta}
	\end{pmatrix}+\begin{pmatrix}
		\bm{A} & \bm{B} \\
		\bm{\Gamma} & \bm{\Delta}
	\end{pmatrix},$$
	with the matrix ${\rm B}^{(1)}$ being of rank 1. Substituting this expression into equations \eqref{BM} and equating the coefficients of powers of $\lambda$ to zero, we find that $\bm{\alpha}$ and $\bm{\delta}$ are independent of $n$ and $m$, and that $\bm{\beta}=\bm{\gamma}=0$. Equation \eqref{BM-a} then yields the following system of equations,
	\begin{subequations}\label{gen-BT-sys}
		\begin{align}
			&\bm{A}+\tilde{\bm{p}}\tilde{\bm{q}}_{10}\bm{\alpha}=\bm{\alpha}\bm{p}\bm{q}_{10}+\bm{A}_{10},\label{gen-BT-sys-a}\\
			&\left(a(n)+\tilde{\bm{p}}\tilde{\bm{q}}_{10}\right)\bm{A}+\tilde{\bm{p}}\bm{\Gamma}=\bm{A}_{10}\left(a(n)+\bm{p}\bm{q}_{10}\right)+\bm{B}_{10}\bm{q}_{10},\label{gen-BT-sys-b}\\
			&\bm{B}+\tilde{\bm{p}}\bm{\delta}=\bm{\alpha}\bm{p},\quad \tilde{\bm{q}}_{10}\bm{\alpha}=\bm{\Gamma}_{10}+\bm{\delta}_{10}\bm{q}_{10},\label{gen-BT-sys-c}\\
			&\left(a(n)+\tilde{\bm{p}}\tilde{\bm{q}}_{10}\right)\bm{B}+\tilde{\bm{p}}\bm{\Delta}=\bm{A}_{10}\bm{p}+\bm{B}_{10},\label{gen-BT-sys-d}\\
			&\tilde{\bm{q}}_{10}\bm{A}+\bm{\Gamma}=\bm{\Gamma}_{10}\left(a(n)+\bm{p}\bm{q}_{10}\right)+\bm{\Delta}_{10}\bm{q}_{10},\label{gen-BT-sys-e}\\
			&\tilde{\bm{q}}_{10}\bm{B}+\bm{\Delta}=\bm{\Gamma}_{10}\bm{p}+\bm{\Delta}_{10}.\label{gen-BT-sys-f}
		\end{align}
	\end{subequations}
	Similarly, equation \eqref{BM-b} yields a system analogous to \eqref{gen-BT-sys}, where all indices ${}_{10}$ are replaced by ${}_{01}$ and $a(n)$ is replaced by $b(m)$.
	
	Since $\rank {\rm B}^{(1)}=1$, we can choose either $\bm{\alpha}\neq 0$ with $\bm{\delta}=0$, or $\bm{\alpha}=0$ and $\bm{\delta}\neq 0$. We consider the former case; the latter can be treated analogously. Assuming $\bm{\alpha}\neq 0$, we can rescale it to $\bm{\alpha} = 1$ without loss of generality. Then equation \eqref{gen-BT-sys-c} yields $\bm{B}=\bm{p}$ and $\bm{\Gamma}=\tilde{\bm{q}}$. Substituting all these into \eqref{gen-BT-sys-f} and its counterpart following from \eqref{BM-b}, we find that $\bm{\Delta}$ must be independent of $n$ and $m$. Without loss of generality, we set $\bm{\Delta} =1$. With these choices, equations 
	\begin{subequations}\label{gen-BT-AY-10}
		\begin{equation}{\small{
					\bm{A}_{10}-\bm{A}=\tilde{\bm{p}}\tilde{\bm{q}}_{10}-\bm{p}\bm{q}_{10},~~ \bm{p}_{10}=\left({\bm{p}}{\bm{q}}_{10} - {\bm{A}}+a(n) \right)\bm{p}+\tilde{\bm{p}},~~ \tilde{\bm{q}}=\tilde{\bm{q}}_{10}\left(\bm{p}\bm{q}_{10}-\bm{A}+a(n)\right)+\bm{q}_{10},}}\label{gen-BT-AY-10-a}
		\end{equation}
		follow from \eqref{gen-BT-sys-a}, \eqref{gen-BT-sys-d}, and \eqref{gen-BT-sys-e}, while equation \eqref{gen-BT-sys-b} is satisfied as a consequence of them. Moreover, equations 
		\begin{equation}{\small{    
					\bm{A}_{01}-\bm{A}=\tilde{\bm{p}}\tilde{\bm{q}}_{01}-\bm{p}\bm{q}_{01},~~ \bm{p}_{01}=\left({\bm{p}}{\bm{q}}_{01} - {\bm{A}}+b(m) \right)\bm{p}+\tilde{\bm{p}},~~ \tilde{\bm{q}}=\tilde{\bm{q}}_{01}\left(\bm{p}\bm{q}_{01}-\bm{A}+b(m)\right)+\bm{q}_{01}}}\label{gen-BT-AY-10-b}
		\end{equation}
	\end{subequations}
	follow from \eqref{BM-b} and the above choices for the elements of matrix $\rm{B}$. 
	
	A consequence of system \eqref{gen-BT-AY-10} is that
	\begin{equation} \label{gen-BT-AY-10-cons}
		({\cal{S}}-1)({\bm{p}} \tilde{\bm{q}}- {\bm{A}}) = [{\bm{p}} \tilde{\bm{q}}_{10}, {\bm{A}} - {\bm{p}} {\bm{q}}_{10}],  ~~~ ({\cal{T}}-1)({\bm{p}} \tilde{\bm{q}}- {\bm{A}}) = [{\bm{p}} \tilde{\bm{q}}_{01}, {\bm{A}} - {\bm{p}} {\bm{q}}_{01}],
	\end{equation}
	which imply that ${\bm{A}} = {\bm{p}} \tilde{\bm{q}} + \epsilon$, $\epsilon \in C({\mathfrak{R}})$. Indeed, with this choice for $\bm{A}$, the second and third relations in \eqref{gen-BT-AY-10-a} and \eqref{gen-BT-AY-10-b} become
	\begin{subequations}\label{gen-BT-AY-10-inter}
		\begin{equation} 
			\bm{p}_{10}=\tilde{\bm{p}}+\left(a(n)- \epsilon\right)\bm{p}(\tilde{\bm{q}}_{10}\bm{p}+1)^{-1},
			\quad \bm{q}_{10}=\tilde{\bm{q}}-\left(a(n)-\epsilon \right)(\tilde{\bm{q}}_{10}\bm{p}+1)^{-1}\tilde{\bm{q}}_{10}
		\end{equation}
		and
		\begin{equation}
			\bm{p}_{01}=\tilde{\bm{p}}+\left(b(m)-\epsilon\right)\bm{p}(\tilde{\bm{q}}_{01}\bm{p}+1)^{-1}, \quad
			\bm{q}_{01}=\tilde{\bm{q}}-\left(b(m)-\epsilon\right)(\tilde{\bm{q}}_{01}\bm{p}+1)^{-1}\tilde{\bm{q}}_{01},
		\end{equation}
	\end{subequations}
	respectively, while the corresponding first relations become identities due to \eqref{gen-BT-AY-10-inter}. Moreover, the commutators on the right-hand sides of \eqref{gen-BT-AY-10-cons} are zero modulo equations \eqref{gen-BT-AY-10-inter}.
	
	In a similar manner, we can analyse the case where $\bm{\alpha}=0$ and $\bm{\delta}\neq 0$. Our analysis can be summarised in the following statement.
	
	\begin{proposition} Let $\rm{B}=\lambda {\rm B}^{(1)} +{\rm B}^{(0)}$, where $\rank {\rm B}^{(1)}=1$, be a Darboux matrix associated with the Lax pair \eqref{nc-Adler-Yamilov-Lax} of the noncommutative Adler--Yamilov system \eqref{nc-Adler-Yamilov-r}. All Darboux matrices of this form, along with their corresponding auto-B\"acklund transformations, fall into one of the following two cases.
		\begin{enumerate}
			\item Darboux matrix has the form 
			\begin{equation}\label{Darboux-matrix-AY}
				\rm{B}=\lambda
				\begin{pmatrix}
					1 & 0 \\
					0 & 0
				\end{pmatrix}
				+
				\begin{pmatrix}
					\bm{p}\tilde{\bm{q}}+\epsilon & \bm{p} \\
					\tilde{\bm{q}} & 1
				\end{pmatrix},
			\end{equation}
			and its entries satisfy the system of difference equations
			\begin{subequations}\label{BT-AY-10}
				\begin{align}
					&  \bm{p}_{10}=\tilde{\bm{p}}+\left(a(n)- \epsilon\right)\bm{p}(\tilde{\bm{q}}_{10}\bm{p}+1)^{-1},
					\quad \bm{q}_{10}=\tilde{\bm{q}}-\left(a(n)-\epsilon \right)(\tilde{\bm{q}}_{10}\bm{p}+1)^{-1}\tilde{\bm{q}}_{10}, \label{BT-AY-10-a}\\
					& \bm{p}_{01}=\tilde{\bm{p}}+\left(b(m)-\epsilon\right)\bm{p}(\tilde{\bm{q}}_{01}\bm{p}+1)^{-1}, \quad
					\bm{q}_{01}=\tilde{\bm{q}}-\left(b(m)-\epsilon\right)(\tilde{\bm{q}}_{01}\bm{p}+1)^{-1}\tilde{\bm{q}}_{01}.\label{BT-AY-10-b}
				\end{align}
			\end{subequations}
			
			\item Darboux matrix has the form
			\begin{equation}\label{Darboux-matrix-AY-2}
				\rm{B}=\lambda
				\begin{pmatrix}
					0 & 0 \\
					0 & 1
				\end{pmatrix}
				+
				\begin{pmatrix}
					\bm{q}\tilde{\bm{p}} + \epsilon& -\tilde{\bm{p}} \\
					-\bm{q} & 1
				\end{pmatrix},
			\end{equation}
			and its entries satisfy the system of difference equations
		\end{enumerate}
		\begin{subequations}\label{BT-AY-01}
			\begin{align}
				& \tilde{\bm{p}}_{10}=\bm{p}+\left(a(n)-\epsilon \right)(1+\tilde{\bm{p}}\bm{q}_{10})^{-1}\tilde{\bm{p}},\quad \tilde{\bm{q}}_{10}=\bm{q}-\left( a(n)-\epsilon \right)\bm{q}_{10}(1+\tilde{\bm{p}}\bm{q}_{10})^{-1},\label{BT-AY-01-a}\\
				& \tilde{\bm{p}}_{01}=\bm{p}+\left( b(m)-\epsilon \right) (1+\tilde{\bm{p}}\bm{q}_{01})^{-1}\tilde{\bm{p}},\quad \tilde{\bm{q}}_{01}=\bm{q}-\left( b(m)-\epsilon \right)\bm{q}_{01}(1+\tilde{\bm{p}}\bm{q}_{01})^{-1}.\label{BT-AY-01-b}
			\end{align}
		\end{subequations}
	\end{proposition}

	\begin{remark}
		Assuming that all the fields in \eqref{BT-AY-10} and \eqref{BT-AY-01} commute, we obtain matrix (27) and system (28) in \cite{FKRX}.
	\end{remark}

	\section{Conclusions}\label{conclusions}
	
	We presented a method for constructing noncommutative integrable systems of difference equations by employing Darboux transformations associated with noncommutative integrable partial differential equations and their corresponding Lax pairs. As an illustrative example, we considered a noncommutative version of the NLS system, specifically system \eqref{NLS-eq-NC}, and its Lax pair \eqref{Lax-NLS-NC}.
	
	In particular, we derived Darboux matrices that are linear in the spectral parameter, along with the associated transformations for the noncommutative NLS system, namely \eqref{eq:T-10} and \eqref{eq:T-01}. The existence of first integrals for systems \eqref{BT-10} and \eqref{BT-01} enabled their reduction to \eqref{eq:T-10-deg} and \eqref{eq:T-01-deg}, respectively. Using these noncommutative Darboux matrices, we constructed a noncommutative discrete integrable NLS-type system, which can be further reduced to a noncommutative Adler--Yamilov-type system, given by \eqref{nc-Adler-Yamilov-r}. Additionally, we constructed a fully discrete noncommutative Toda-type equation, namely \eqref{noncomm-d-Toda}. Furthermore, by suitably generalising the approach introduced in \cite{FKRX}, we derived Darboux and B\"acklund transformations for the noncommutative Adler--Yamilov system \eqref{nc-Adler-Yamilov-r}.
	
	It is well known that Darboux and B\"acklund transformations can be employed to construct solutions of associated integrable systems. In particular, it would be interesting to construct solutions of system \eqref{nc-Adler-Yamilov-r}, starting from a simple seed solution and applying either B\"acklund transformation \eqref{BT-AY-10} or \eqref{BT-AY-01}. This methodology was demonstrated in \cite{FKRX} for the commutative version of system \eqref{nc-Adler-Yamilov-r}, and worth exploring whether the resulting soliton solutions can be extended to the noncommutative setting.
	
	Darboux and B\"acklund transformations may involve additional functions (potentials), such as $\bm{f}$ in \eqref{BT-10} and $\bm{g}$ in \eqref{BT-01}. However, there exist Darboux matrices that depend only on a potential, with the corresponding auto-B\"acklund transformations also expressed using the potential. Such transformations were studied in \cite{PaulX-4}, where an auto-B\"acklund transformation for the noncommutative discrete KdV equation was derived. It was further shown that the associated superposition principle gives rise to a noncommutative Yang--Baxter map.
	
	Noncommutative Yang--Baxter and tetrahedron maps have attracted considerable attention (see, for instance, \cite{Doliwa-Kashaev, Kassotakis-Kouloukas} and the references therein). The Darboux matrix \eqref{DT2-10} can be used to construct a noncommutative Yang--Baxter maps \cite{Sokor-Nikitina}, while the Darboux matrix \eqref{DT-10-deg} can be used to derive a Zamolodchikov tetrahedron map \cite{Sokor-2022}, which can be restricted to a noncommutative version of Sergeev's map \cite{Kashaev-Sergeev}. There is also growing interest in pentagon maps and their connection to matrix refactorisation problems \cite{Dimakis-Hoissen-2015, Kassotakis-2023}. Recent classification results on pentagon maps appeared in~\cite{EKT-2024}, where the first examples of noncommutative quadrirational pentagon maps were introduced. It would be interesting to further explore the connection between matrix refactorisation problems for noncommutative Darboux matrices and pentagon maps.
	
	The existence of a Lax pair can be regarded as a criterion for integrability and, as we demonstrated, may be employed to derive other integrability properties of the associated system. In particular, a Lax pair can be used to obtain symmetries of the corresponding integrable system. It was shown in \cite{PaulX-4} that system \eqref{nc-Adler-Yamilov-r} admits hierarchies of symmetries in both lattice directions. For commutative systems of difference equations, symmetries can also be derived using the theory of integrability conditions (see, for example, \cite{BX,PaulX, Sasha-Pavlos, PaulX-2}). This theory has recently been extended to the study of noncommutative differential-difference equations in \cite{NW}, where a related classification problem was also addressed. It would be of interest to further develop the theory of integrability conditions for noncommutative difference equations, in analogy with the commutative case.
	
	\section{Acknowledgements}
	The authors thank the London Mathematical Society for funding the visit of Sotiris Konstantinou-Rizos to the UK (Grant {\emph{Research in Pairs}}, Project No. 41962). The work of Sotiris Konstantinou-Rizos was supported by a development programme for the Regional Scientific and Educational Mathematical Centre of the Yaroslavl State University with financial support from the Ministry of Science and Higher Education of the Russian Federation (Agreement No. 075-02-2025-1636). We would like to thank S. Igonin for useful discussions.


\begin{thebibliography}{99}
		
		\bibitem{BX} L. Brady, P. Xenitidis (2025) Systems of difference equations, symmetries and integrability conditions {\it{to appear in Theor. Math. Phys.}} 
		
		\bibitem{Dimakis-Hoissen-2015}
		A. Dimakis, F. M\"uller-Hoissen (2015) {Simplex and polygon equations} {\it{SIGMA}} {\bf{11}} {042}.
		
		\bibitem{Doliwa-Kashaev}
		A. Doliwa, R.M. Kashaev (2020) Non-commutative birational maps satisfying Zamolodchikov equation, and Desargues
		lattices {\em J. Math. Phys.} {\textbf{61}} 092704
		
		\bibitem{EKT-2024}
		Ch. Evripidou, P. Kassotakis, A. Tongas (2024) On quadrirational pentagon maps {\emph{J. Phys. A: Math. Theor.}} {\bf{57}} 455203
		
		\bibitem{FKRX}
		X. Fisenko, S. Konstantinou-Rizos, P. Xenitidis (2022) A discrete Darboux-Lax scheme for integrable difference equations
		{\em Chaos, Solitons \& Fractals} {\textbf{158}} 112059
		
		
		\bibitem{Hiet-Frank-Joshi}
		J. Hietarinta, N. Joshi, F. Nijhoff (2016) Discrete Systems and Integrability, Cambridge texts in applied mathematics, Cambridge University Press
		
		\bibitem{Hirota}
		R. Hirota (1977) Nonlinear Partial Difference Equations. I. A Difference Analogue of the Korteweg-de Vries Equation
		{\em J. Phys. Soc. Japan} {\textbf{43}} 1424--1433
		
		\bibitem{Kashaev-Sergeev} 
		R.M. Kashaev, I.G. Koperanov, S.M. Sergeev (1998) Functional Tetrahedron Equation {\emph{Theor. Math. Phys.}} {\bf{117}} 370--384
		
		\bibitem{Kassotakis-2023}
		P. Kassotakis (2023) Matrix factorizations and pentagon maps {\em Proc. R. Soc. A.} {\textbf{479}} 20230276
		
		
		\bibitem{Kassotakis-Kouloukas}
		P. Kassotakis, T. Kouloukas (2022) On non-abelian quadrirational Yang-Baxter maps {\em J. Phys. A: Math. Theor.} {\textbf{55}} 175203 
		
		\bibitem{Sokor-2022}
		S. Konstantinou-Rizos (2022) Noncommutative solutions to Zamolodchikov's tetrahedron equation and matrix six-factorisation problems
		{\emph{Physica D}} {\bf{440}} 133466 
		
		\bibitem{SPS}
		S. Konstantinou-Rizos, A.V. Mikhailov, P. Xenitidis (2015) Reduction groups and related integrable difference systems of nonlinear Schr\"odinger type {\em J. Math. Phys.} {\textbf{56}} 082701
		
		\bibitem{Sokor-Nikitina}
		S. Konstantinou-Rizos, A.A. Nikitina (2024) Yang--Baxter maps of KdV, NLS and DNLS type on division rings {\emph{Physica D}} {\bf{465}}  134213
		
		\bibitem{Manakov}
		S.V. Manakov (1974) On the theory of two-dimensional stationary self focussing of electromagnetic waves {\em Sov. Phys. JETP} {\textbf{38}(2)} 248--253
		
		\bibitem{PaulX}
		A.V. Mikhailov, J.-P. Wang, P. Xenitidis (2011) Recursion operators, conservation laws and integrability conditions for difference equations {\emph{Theor. Math. Phys.}} {\textbf{167}} 421--443
		
		\bibitem{Sasha-Pavlos}
		A.V. Mikhailov, P. Xenitidis (2014) Second Order Integrability Conditions for Difference Equations: An Integrable Equation {\emph{Lett. Math. Phys.}} {\textbf{104}} 431--450
		
		\bibitem{NW} V. Novikov, J.P. Wang (2025) Integrability of Nonabelian Differential--Difference Equations: The Symmetry Approach {\emph{Commun. Math. Phys.}} {\bf{406}} 11
		
		\bibitem{OS}
		P. J. Olver, V. V. Sokolov. (1998) Integrable Evolution Equations on Associative Algebras {\emph{Commun. Math. Phys.}} {\textbf{193}} 245--268
		
		\bibitem{Suris}
		Yu. Suris. (1995) Bi-Hamiltonian structure of the qd algorithm and new discretizations of the Toda lattice {\em Phys. Lett. A}
		{\bf{206}} 153--161
		
		\bibitem{PaulX-2}
		P. Xenitidis (2018) Determining the symmetries of difference equations {\emph{Proc. R. Soc. A.}} {\bf{474}} 20180340
		
		\bibitem{PaulX-4}
		P. Xenitidis (2025) Noncommutative discrete equations, symmetries and reductions {\em arXiv:2507.04472}
		
	\end{thebibliography}
\end{document}